\documentclass[11pt,english]{article}
\usepackage[latin9]{inputenc}
\usepackage[margin=1in]{geometry}
\usepackage{amsthm}
\usepackage{amsmath,amssymb}
\usepackage[usenames,dvipsnames]{xcolor}

\newcommand{\ignore}[1]{}%

\newif\ifFULL
\FULLtrue

\makeatletter
  \theoremstyle{definition}
  \newtheorem{defn}{\protect\definitionname}
  \theoremstyle{plain}
  \newtheorem{conjecture}{\protect\conjecturename}
  \theoremstyle{plain}
  \newtheorem*{thm*}{\protect\theoremname}
  \theoremstyle{plain}
  
  \theoremstyle{remark}
  
  \theoremstyle{plain}
  \newtheorem{lem}{\protect\lemmaname}
\theoremstyle{plain}
\newtheorem{thm}{\protect\theoremname}
\theoremstyle{plain}
  \newtheorem*{lem*}{\protect\lemmaname}

\def\epsilon{\varepsilon}

\usepackage{complexity,color,bm}
\usepackage{colortbl}

\newcommand{\e}{\mathrm{e}}

\definecolor{gray-comment}{gray}{0.5}

\protected\edef\mathbb{%
  \unexpanded\expandafter\expandafter\expandafter{%
    \csname mathbb \endcsname
  }%
}

\theoremstyle{plain}

\makeatletter
\newtheorem*{rep@theorem}{\rep@title}
\newcommand{\newreptheorem}[2]{%
\newenvironment{rep#1}[1]{%
 \def\rep@title{#2 \ref{##1}}%
 \begin{rep@theorem}}%
 {\end{rep@theorem}}}
\makeatother

\newreptheorem{rcor}{Corollary}
\newreptheorem{rthm}{Theorem}

\usepackage{thmtools}
\usepackage{thm-restate}
\makeatother

\usepackage{babel}
  \providecommand{\claimname}{Claim}
  \providecommand{\conjecturename}{Conjecture}
  \providecommand{\definitionname}{Definition}
  \providecommand{\lemmaname}{Lemma}
  \providecommand{\theoremname}{Theorem}
\providecommand{\corollaryname}{Corollary}
\providecommand{\theoremname}{Theorem}

\begin{document}

\author{Yakov Babichenko\thanks{Technion. Email:yakovbab@tx.technion.ac.il}
\and Christos Papadimitriou\thanks{University of California, Berkeley. Email: christos@cs.berkeley.edu}
\and Aviad Rubinstein\thanks{University of California, Berkeley. Email: aviad@cs.berkeley.edu}
}

\title{\bf Can Almost Everybody be Almost Happy?\\
PCP for PPAD and the Inapproximability of Nash}
\date{}

%\author{Christos H. Papadimitriou}
%\thanks{University of California, Berkeley} 

\thispagestyle{empty}\maketitle\setcounter{page}{0}

\begin{abstract}\noindent
We conjecture that \PPAD\ has a PCP-like complete problem, seeking a near equilibrium in which all but very few players have very little incentive to deviate.  We show that, if one assumes that this problem requires exponential time, several open problems in this area are settled.  The most important implication, proved via a ``birthday repetition'' reduction,  is that the $n^{O(\log n)}$ approximation scheme of Lipton et al.~\cite{LMM03_quasi_poly} for the Nash equilibrium of two-player games {\em is essentially optimum.}   Two other open problems in the area are resolved once one assumes this conjecture, establishing that certain approximate equilibria are \PPAD-complete:  Finding a {\em relative} approximation of two-player Nash equilibria (without the well-supported restriction of \cite{Das13_multiplicative_hardness}), and an approximate competitive equilibrium with equal incomes \cite{ACEEI_Bud11} with small clearing error and near-optimal Gini coefficient.
\end{abstract}

\thispagestyle{empty}
\newpage

\section{Introduction}
% !TeX root = PCP_for_PPAD2.tex

It is known that finding an $\epsilon$-approximate Nash equilibrium in a two-person game:
\begin{enumerate}
\item is \PPAD-complete if $\epsilon$ is inversely polynomial \cite{2-player_nash_CDT09};
\item but can be solved in quasipolynomial time for any fixed $\epsilon>0$ \cite{LMM03_quasi_poly}; and
\item the smallest known polynomially attainable approximation ratio is still over $3\over 10$ \cite{TS08-approx_nash}.  
\end{enumerate}
These three facts articulate rather dramatically the mystery, by now almost a decade old, of the problem's approximability.  

Can the inversely polynomial inapproximability bound of \cite{2-player_nash_CDT09} be improved to constant?  Unlikely, because by \cite{LMM03_quasi_poly} this would imply a quasipolynomial algorithm for the iconic problem in PPAD:
\begin{defn}{\sc End of the Line}: (\cite{NASH-is-PPAD-hard_DGP09}) 
Let $S$ and $P$ be two circuits with (at most) $\tilde{O}(n)$ gates each (computing the predecessor and successor correspondences) with $n$ input bits and $n$ output bits each, such that $P\left(0^{n}\right)=0^{n}\neq S\left(0^{n}\right)$, find an input $x\in\left\{ 0,1\right\} ^{n}$ such that $P\left(S\left(x\right)\right)\neq x$ or $S\left(P\left(x\right)\right)\neq x\neq0^{n}$.
\end{defn}
Conventional wisdom, supported by black-box lower bounds \cite{HPV89,BCEIP98-relativized}, is that this problem requires exponential time to solve.  In direct analogy with the exponential time hypothesis \cite{IP01-ETH}, we state the following:

\begin{conjecture}
\label{conj:eth}
{\sc End of the Line} requires $2^{\tilde\Omega(n)}$ time.
\end{conjecture}

It was recently shown that this conjecture implies an identical lower bound for the {\sc Gcircuit} problem \cite{2-player_nash_CDT09,NASH-is-PPAD-hard_DGP09}: Recall
that a standard arithmetic circuit has inputs, outputs, and gates
that define arithmetic operations on its lines. In a generalized circuit there are no outputs, and
the dependencies between values become cyclical, thereby inducing a
constraint satisfaction problem: the variables correspond to the values
on the lines, and each gate induces a constraint on the incoming
and outgoing lines. The $\epsilon$-{\sc Gcircuit} is defined to be the problem of finding values on the lines that $\epsilon$-approximately satisfy all the 
constraints induced by the arithmetic gates.  It was recently shown \cite{Rub14-polymatrix-Nash} that $\epsilon$-{\sc Gcircuit} is \PPAD-complete for some small constant $\epsilon$, via $\tilde O(n)$ reductions.  Therefore, Conjecture 1 implies that, for some small $\epsilon>0$, roughly exponential time is required to find a solution to this problem.
There are several known ways to further reduce the $\epsilon$-{\sc Gcircuit} problem to the two-player Nash equilibrium problem; however, all of these reductions fail to preserve the quality of approximation, and, as we pointed out above, it is likely that no approximation-preserving reductions are possible, as this would imply a quasipolynomial-time algorithm for {\sc End of the Line} through the algorithm in \cite{LMM03_quasi_poly}, contradicting Conjecture 1.

{\em In this paper, we identify a plausible new conjecture, a strengthening of Conjecture 1 in a natural and novel direction, which implies that the quasipolynomial approximation algorithm of \cite{LMM03_quasi_poly} is optimum.}   In particular, we define the $(\varepsilon,\delta)$-{\sc Gcircuit} to be the problem of finding values for the variables (lines) that $\varepsilon$-approximately satisfy a fraction {\em of at least $(1-\delta)$ of the constraints (gates).}

\begin{conjecture}\label{conj:pcp}
There exist constants $\epsilon,\delta>0$
such that there is a quasilinear reduction from {\sc End of the Line} to the problem $(\epsilon,\delta)$-{\sc Gcircuit}.  Therefore, also assuming Conjecture 1, $(\epsilon,\delta)$-{\sc Gcircuit} requires $2^{\tilde\Omega(n)}$ time.
\end{conjecture}

As we have mentioned, this latter statement, albeit with $\delta=0$, is known to follow from only Conjecture 1, for some $\epsilon>0$ \cite{Rub14-polymatrix-Nash}.

\subsection*{The Connection to PCP}  
The $(\epsilon,\delta)$-{\sc Gcircuit} problem is a constraint satisfaction problem: each line/mixed strategy is a real variable, and each gate/player defines a constraint. Since we are considering $\epsilon$-approximate satisfaction of the constraints, each variable need be represented using only
a constant (depending on $\epsilon$) number of bits. Therefore, a satisfying assignment (truncations of a Nash equilibrium) can be distinguished from an unsatisfying one (one violating the $(\epsilon-\delta)$-relaxation) by querying a constant (proportional to $1/\delta$) number of bits of the input, determined at random (but of course not uniformly so). This suggests that Conjecture \ref{conj:pcp} can be interpreted as a probabilistically checkable proof formulation of {\sc Gcircuit}: informally, it states that ``\PPAD\ has a PCP'' (that is, a complete problem whose witnesses can be verified by examining, at random, a finite number of bits).

The main result of this paper, explained next, reveals another similarity between the PCP formulation of NP and the $(\epsilon,\delta)$-{\sc Gcircuit} problem:  The intractability of the latter problem implies, among several other inapproximability theorems, the strongest possible inapproximability result for the 2-player Nash equilibrium problem, arguably the central open question in the area.

%Actually, we will need a stronger conjecture which requires the existence of a \emph{quasi-linear} reduction from {\sc EndOfTheLine} to $(\epsilon,\delta)$-{\sc Gcircuit} (note that such a quasi-linear reduction from {\sc EndOfTheLine} to the simpler problem of $\varepsilon$-{\sc Gcircuit} exists, see \cite{Rub14-polymatrix-Nash}).

%\begin{conjecture}[Quasilinear PCP for \PPAD]\label{conj:quasi-pcp}
%There exist constants $\epsilon,\delta>0$ such that the {\sc EndOfTheLine} problem over a graph of size $2^n$ can be reduced to the problem $(\varepsilon,\delta)$-{\sc Gcircuit} with $\tilde{O}(n)$ gates\footnote{Actually, a weaker conjecture that assumes that {\sc EndOfTheLine} over a graph of size $2^n$ can be reduced to the problem {\sc Gcircuit} with $\tilde{O}(n^\alpha)$ gates for $\alpha<2$, is sufficient to derive similar negative results to those in Corollaries \ref{cor:quasi-poly-PPAD-graph} and \ref{cor:quasi-poly-time}: the lower bound of $m^{\tilde{\Omega}(\log m)}$ will be replaced by $m^{\tilde{\Omega}\left( \log^{2/ \alpha-1} m \right)}$.}.
%\end{conjecture}

\subsection*{The Main Result}
We denote by $\varepsilon$-{\sc 2Nash} the problem of finding an $\epsilon$-Nash equilibrium in a bimatrix game. Our main result %in this paper is %a $2^{\tilde{O}\left(\sqrt{n}\right)}$-reduction from $(\varepsilon,\delta)$-{\sc Gcircuit} to $\epsilon'$-{\sc 2Nash}.
is the following:

\begin{thm}
\label{thm:main}
There is an $\varepsilon >0$ such that, assuming Conjecture \ref{conj:pcp}, solving $\varepsilon$-{\sc 2Nash} for two-player games with $n$ strategies requires $n^{\tilde\Omega(\log n)}$ time.  
\end{thm}

Our proof, given in Section 3, employs the technique of ``birthday repetition,'' pioneered by \cite{AIM14-birthday} and used in different contexts by \cite{BRKW15-DkS}, and in particular by \cite{BKW15-best_nash} to show intractability of a Nash equilibrium-related problem.  Starting from a polymatrix game with two strategies per player and in which the player interactions form a cubic bipartite graph with $n$ nodes on each side, the players on both sides are broken into blocks of size $\sqrt{n}$.  The game is simulated by a two-player game, in which each player simulates the nodes in one of the sides of the bipartite graph by choosing a block and a strategy for each node in it; that is, the total number of actions of each of the two players is about $2^{\sqrt{n}}$, and such is the complexity of the reduction (this is necessary if one wishes to rule out better than quasipolynomial algorithms). The interactions between blocks, plus certain particular side games played by the two players, ensure a faithful enough simulation of the original multimatrix game.

\paragraph{Remark:} It is not hard to see that a quasipolynomial lower bound follows from a weaker assumption than Conjecture 2, namely that the $(\epsilon,\delta)$-{\sc Gcircuit} problem requires $2^{\tilde\Omega(n^{\alpha})}$ time, {\em for some $\alpha > {1\over 2}$.}  

In addition to Theorem 1, we prove two other complexity consequences of Conjecture 2, solving certain open problems in the area:  An improved inapproximability result for the problem of relative (multiplicative) approximation of the Nash equilibrium first established by Dskalakis \cite{Das13_multiplicative_hardness}, and an inapproximability result for the problem of finding a competitive equilibrium with equal incomes and indivisible goods \cite{ACEEI_Bud11,OPR14_A-CEEI} when one seeks to minimize the Gini index of the income distribution.

\section{The $(\epsilon,\delta)$-{\sc Gcircuit} and Weak Approximate Nash }
\label{sec:weakNash}

Generalized circuits are similar to the standard algebraic circuits,
the main difference being that generalized circuits contain cycles,
which allow them to verify fixed points of continuous functions. We
restrict the class of generalized circuits to include only a particular
list of gates described below. Formally,
\begin{defn}
{[}{Generalized circuits, \cite{2-player_nash_CDT09}}{]} A {\em
generalized circuit} ${\cal S}$ is a pair $\left(V,{\cal {T}}\right)$,
where $V$ is a set of nodes and ${\cal {T}}$ is a collection of
gates. Every gate $T\in{\cal {T}}$ is a 5-tuple $T=G\left(\zeta\mid v_{1},v_{2}\mid v\right)$,
in which $G\in\{G_{\zeta},G_{\times\zeta},G_{=},G_{+},G_{-},G_{<},G_{\vee},G_{\wedge},G_{\neg}\}$
is the type of the gate; $\zeta\in\mathbb{R}\cup\left\{ nil\right\} $
is an (optional) real parameter; $v_{1},v_{2}\in V\cup\{nil\}$ are the first
and second input nodes of the gate (one or both of them may be missing); and $v\in V$ is the output node; no two distinct gates have the same output. 
\end{defn}

Alternatively, we can think of each gate as a constraint on the values
on the incoming and outgoing wires. We are interested in the following
constraint satisfaction problem: given a generalized circuit, find
an assignment to all the wires that simultaneously satisfies all the
gates. When every gate computes a continuous function of the incoming
wires, a solution must exist by Brouwer's fixed point theorem. 

In particular, we are interested in the approximate version of this
CSP, where we must approximately satisfy most of the constraints.
\begin{defn}
Given a generalized circuit ${\cal S}=\left(V,{\cal T}\right)$, we
say that an assignment $\mathrm{\mathbf{x}}\colon V\rightarrow\left[0,1\right]$
{\em $(\epsilon, \delta)$-approximately satisfies} ${\cal {S}}$ if for
all but a $\delta$-fraction of the gates, $\mathbf{x}$ satisfies the corresponding
constraints:\\

\begin{tabular}{|c|c|}
\hline 
Gate & Constraint\tabularnewline
\hline 
\hline 
$G_{\zeta}\left(\alpha\mid\mid a\right)$ & $\mathbf{x}\left[a\right]=\alpha\pm\epsilon$\tabularnewline
\hline 
$G_{\times\zeta}\left(\alpha\mid a\mid b\right)$ & $\mathbf{x}\left[b\right]=\alpha\cdot\mathbf{x}\left[a\right]\pm\epsilon$\tabularnewline
\hline 
$G_{=}\left(\mid a\mid b\right)$ & $\mathbf{x}\left[b\right]=\mathbf{x}\left[a\right]\pm\epsilon$\tabularnewline
\hline 
$G_{+}\left(\mid a,b\mid c\right)$ & $\mathbf{x}\left[c\right]=\min\left(\mathbf{x}\left[a\right]+\mathbf{x}\left[b\right],1\right)\pm\epsilon$\tabularnewline
\hline 
$G_{-}\left(\mid a,b\mid c\right)$ & $\mathbf{x}\left[c\right]=\max\left(\mathbf{x}\left[a\right]-\mathbf{x}\left[b\right],0\right)\pm\epsilon$\tabularnewline
\hline 
$G_{<}\left(\mid a,b\mid c\right)$ & $\mathbf{x}\left[c\right]=\begin{cases}
1\pm\epsilon & \mathbf{x}\left[a\right]<\mathbf{x}\left[b\right]-\epsilon\\
0\pm\epsilon & \mathbf{x}\left[a\right]>\mathbf{x}\left[b\right]+\epsilon
\end{cases}$\tabularnewline
\hline 
$G_{\vee}\left(\mid a,b\mid c\right)$ & $\mathbf{x}\left[c\right]=\begin{cases}
1\pm\epsilon & \mathbf{x}\left[a\right]=1\pm\epsilon\mbox{ or }\mathbf{x}\left[b\right]=1\pm\epsilon\\
0\pm\epsilon & \mathbf{x}\left[a\right]=0\pm\epsilon\mbox{ and }\mathbf{x}\left[b\right]=0\pm\epsilon
\end{cases}$\tabularnewline
\hline 
$G_{\wedge}\left(\mid a,b\mid c\right)$ & $\mathbf{x}\left[c\right]=\begin{cases}
1\pm\epsilon & \mathbf{x}\left[a\right]=1\pm\epsilon\mbox{ and }\mathbf{x}\left[b\right]=1\pm\epsilon\\
0\pm\epsilon & \mathbf{x}\left[a\right]=0\pm\epsilon\mbox{ or }\mathbf{x}\left[b\right]=0\pm\epsilon
\end{cases}$\tabularnewline
\hline 
$G_{\neg}\left(\mid a\mid b\right)$ & $\mathbf{x}\left[b\right]=\begin{cases}
1\pm\epsilon & \mathbf{x}\left[a\right]=0\pm\epsilon\\
0\pm\epsilon & \mathbf{x}\left[a\right]=1\pm\epsilon
\end{cases}$\tabularnewline
\hline 
\end{tabular}\\

Given a generalized circuit ${\cal S}=\left(V,{\cal T}\right)$, \textsc{$(\epsilon, \delta)$-Gcircuit}
is the problem of finding an assignment that $(\epsilon, \delta)$-approximately
satisfies it. 
\end{defn}

%\section{Weak approximate Nash equilibrium}
% !TeX root = PCP_for_PPAD2.tex

\begin{defn}
In an {\em $\left(\epsilon,\delta\right)$-weak approximate Nash equilibrium}
at most a $\delta$-fraction of the players can gain more than $\epsilon$
by deviating.
\end{defn}

\begin{thm}
Conjecture \ref{conj:pcp} is equivalent to the following statement:
There exist constants $\epsilon', \delta' >0$, such that finding an 
$\left(\epsilon',\delta'\right)$-weak approximate Nash equilibrium
in a polymatrix graphical game with degree 3 and 2 actions per player
requires $2^{\tilde \Omega (n)}$ time.
\end{thm}

\begin{proof}
The reduction from $(\epsilon,\delta)$-{\sc Gcircuit} to 
$\left(\epsilon',\delta'\right)$-weak approximate Nash equilibrium
follows analogously to the reduction in \cite{Rub14-polymatrix-Nash} 
from $\epsilon$-{\sc Gcircuit} to $\epsilon'$-approximate Nash equilibrium.
%{\color{blue} {\bf NEW TEXT}}
However there are subtleties in the first step of the reduction%
\footnote{We thank Mark Braverman for bringing these issues to our attention.} that goes from an arbitrary instance of {\sc Gcircuit} to one where the fan-out is at most $2$. We further break this step into three sub-steps.

%The first step is to reduce $(\epsilon,\delta)$-{\sc Gcircuit} to $(\epsilon'',\delta'')$-{\sc Gcircuit} with fan-out $2$. 
%We do this in three sub-steps. 

\subsubsection*{Constant fan-out gates}

In the first sub-step, we get rid of gates with super-constant fan-out, i.e.~we reduce to $(\epsilon,\delta)$-{\sc Gcircuit} to $(\epsilon, \delta/4)$-{\sc Gcircuit} with fan-out $O(1/\delta)$.
Let ${\mathcal{S}}^0$ denote the original instance. 
Note that if we have $n$ gates, there are at most $2n$ incoming wires. Therefore, at most $\delta/4$-fraction of gates have $8/\delta$ outgoing wires or more. We create a new instance ${\mathcal{S}}^1$ by removing each such gate $g$ with $d_g \ge 8/\delta$ outgoing wires, and replacing it with $d_g$ constant gates $G_{\zeta}$, with an arbitrary constant $\alpha$, and connect each one to a different outgoing wire; we discard $g$'s original incoming wire(s). This sub-step introduces at most $2n$ new gates. 

\paragraph{Analysis of reduction to constant fan-out:}
We claim that any $(\epsilon, \delta/4)$-approximate solution for the ${\mathcal{S}}^1$ can be mapped to an $(\epsilon, \delta)$-approximate solution for ${\mathcal{S}}^0$: Given an $(\epsilon, \delta/4)$-approximate solution for ${\mathcal{S}}^1$, we first transform it so that all constraints on the new constant gates are satisfied: each of those new constant gates $g$ feeds into one old gate $g'$; for each constraint that was originally unsatisfied for $g$, it's possible that now the constraint on $g'$ is unsatisfied, but the total number of unsatisfied gates does not increase. We now obtain an assignment for the original circuit (setting the high-degree gates' assignment to $\alpha$). Other than the high-degree gates (which contribute only a $\delta/4$-fraction of gates), every gate satisfied in the approximate solution for the ${\mathcal{S}}^1$ is satisfied in ${\mathcal{S}}^0$; the {\em fraction} of unsatisfied gates increases by at most a $3$ factor because the total number of gates decreases back to $n$.

\subsubsection*{Multiplication gates with bounded factor}
Observe that for any gate $g$ that is not a $G_{\times \zeta}$-gate, if $\mathrm{\mathbf{x}}$ $\epsilon$-satisfies $g$ and $||\mathrm{\mathbf{x}}' - \mathrm{\mathbf{x}}||_{\infty} \le \epsilon$, then $\mathrm{\mathbf{x}}'$ $4\epsilon$-satisfies $g$. 

Technically, the multiplication gates $G_{\times \zeta}$ can be instantiated with an arbitrarily large factor $\alpha$, which forces them to compute a function with an arbitrarily large Lipschitz constant%
\footnote{This issue is completely brushed under the rug in earlier works of~\cite{2-player_nash_CDT09,Rub14-polymatrix-Nash}, but that is easily fixable because they give concrete \PPAD-completeness results for {\sc Gcircuit} that can be verified to use bounded factors. Here, in contrast, we start with the {\em conjecture} that {\sc Gcircuit} is hard for some parameters.}; this in turn will be inconvenient in our next sub-step when we go from constant fan-out to fan-out $2$. 
We call such multiplication gates with factor $\alpha > \alpha^* = \epsilon^{-O(1/\delta)}$ {\em unstable}.

Given an instance ${\mathcal{S}}^1$ of $(\epsilon, \delta/4)$-{\sc Gcircuit}, we first replace all the unstable gates with constant-$1$ gates. 
Intuitively, for any gate $g$ that feeds into an unstable gate in ${\mathcal{S}}^1$, we want to round up its output without breaking any other gates that receive their inputs from $g$.
Now we can use the fact that $g$ has bounded fan-out: by pigeonhold principle, there must exist an interval $(\epsilon \gamma/2, \gamma)$ for $\gamma \in (4,  \epsilon^{-O(1/\delta)})$ such that $g$ does not feed into any $G_{\times \zeta}$ gates with factor in that interval.
We also replace all other gates with factor $> \gamma$ with constant-$1$ gates.  
Denote the resulting instance by ${\mathcal{S}}^2$

\paragraph{Analysis of reduction to gates with bounded factor:}
Suppose we have an $(\epsilon/2, \delta/4)$-approximate solution ${\mathcal{S}}^2$, we can almost immediately convert it to an approximate solution for ${\mathcal{S}}^1$, with the caveat that the value on incoming wires to unstable gates cannot be very close to $0$ (because that would not be consistent with the output of the corresponding constant-$1$ gate in the in ${\mathcal{S}}^2$). 

For any gate $g$ that feeds into an unstable gate (in ${\mathcal{S}}^1$), let $\gamma$ be as in the construction. If $g$'s value in the approximate solution for the ${\mathcal{S}}^2$ is $<1/\gamma$, we round it up to $1/\gamma$ (otherwise we leave it unchanged). 
Notice that the solution is now consistent with all gates that were replaced by constant-$1$ gates. For the remaining multiplication gates that receive their inputs from $g$, their factor is at most $\epsilon \gamma/2$, so their intended output increases by at most an additive $\epsilon/2$. Hence all gates that are $\epsilon/2$-satisfied by the approximate solution in ${\mathcal{S}}^2$ are $\epsilon$-satisfied by the adjusted solution for ${\mathcal{S}}^1$.

\subsubsection*{Fan-out $2$}
%Let $\epsilon''  = \epsilon^{O(1/\delta)}$ be such that in the previous sub-step we guarantee that the multiplication gate factors are bounded by $1/\sqrt{\epsilon''}$.
We can now reduce from $(\epsilon/2, \delta/4)$-{\sc Gcircuit} with fan-out $O(1/\delta)$ and $\alpha^*$-bounded multiplication gates to $(\epsilon'' = \Theta(\epsilon / (\alpha \log(1/\delta))), \delta/4)$-{\sc Gcircuit} with fan-out $2$. Here we follow~\cite{NASH-is-PPAD-hard_DGP09}, and augment each degree $2 <  d = O(1/\delta)$  gate with larger fan-out with a depth-$O(\log(d))$ binary tree of $G_{=}$ gates with fan-out $2$. Let ${\mathcal{S}}^3$ denote this new fan-out $2$ instance.

\paragraph{Analysis of reduction to fan-out $2$:}
We claim that any $(\epsilon'', O(\delta^2))$-approximate solution to ${\mathcal{S}}^3$, immeiately gives an $(O(\epsilon), O(\delta))$-approximate solution for ${\mathcal{S}}^2$. For every new $G_{=}$ gate that is more than $O(\epsilon/\log(1/\delta))$-unsatisfied in ${\mathcal{S}}^3$, we can assign an arbitrary value to the wires corresponding to the entire binary tree $G_{=}$ gates in ${\mathcal{S}}^2$; this increases the fraction of unsatisfied gates by a factor of at most $O(1/\delta)$. 

For any binary tree where all the $G_{=}$ gates are $\epsilon''$-satisfied, the inputs to the gates at the leaves of the trees change by at most $O(\epsilon'' \log(1/\delta))$. Therefore, keeping the same output for all the gates at the leaves is still an $\epsilon$-satisfying assignment for the same gates in ${\mathcal{S}}^2$.

\subsubsection*{Weak approximate Nash equilibrium}
We now reduce $(\epsilon,\delta)$-{\sc Gcircuit} with fan-out $2$ 
to $\left(\epsilon',\delta'\right)$-weak approximate Nash equilibrium.
Daskalakis et al.~construct, for each gate in the generalized circuit,
a game gadget with a few players whose mixed strategies at approximate equilibrium simulate the computation carried by the gate.
In other words, every approximate equilibrium of the gate gadget 
induces an approximately satisfying assignment to the corresponding input and output lines in the generalized circuit.
See Lemma \ref{lem:times-gadget} for an example of such a gadget. 
The gadgets are concatenated by identifying the ``output player'' of one gadget with the ``input player(s)" of the next gadget(s) in the generalized circuit.
Since each gadget is composed of only a constant number of players, 
if all but a $\delta'$-fraction of the players play approximately at equilibrium,
all but a $\delta = \Theta(\delta')$-fraction of the gates are violated by the induced assignment to the generalized circuit.

The above construction suffices to show hardness for an $\left(\hat \epsilon,\delta'\right)$-weak well-supported Nash equilibrium,
i.e. one for where all but a $\delta'$-fraction of the players,
every action in the support is $\hat \epsilon$-approximately best response.
Finally, for constant degree graphical games, it is not hard to derive
an $\left(\hat \epsilon,\delta'\right)$-weak well-supported Nash equilibrium
from any $\left(\epsilon',\delta'\right)$-weak approximate Nash equilibrium for $\epsilon' = \Theta(\hat \epsilon^2)$ \cite[Lemma~4]{Rub14-polymatrix-Nash}.\\

In the other direction, from 
$\left(\epsilon',\delta'\right)$-weak approximate Nash equilibrium
to $(\epsilon,\delta)$-{\sc Gcircuit}, 
it suffices to construct a generalized circuit that computes an $\epsilon'$-approximate best response of each player (except a $\delta'$-fraction).
Notice that in a graphical game with a constant number of actions per player, 
given a profile of mixed strategies, one only needs a constant number of arithmetic gates to compute a best response of any player.
In particular, to compute an $\epsilon$-approximate best response
it suffices to use gates of finite precision $\epsilon' = \Theta(\epsilon)$.
Finally, since the number of gates per player is constant,
if at most a $\delta$-fraction of the gates are $\epsilon$-unsatisfied, 
they can appear in the best-response computation for at most a $\delta'=\Theta(\delta)$-fraction of the players.
\end{proof}

\begin{lem}
\label{lem:times-gadget}
($G_{\times\zeta}$ gadget, \cite{NASH-is-PPAD-hard_DGP09})

Let $v_{1}$,$v_{2}$, and $w$ be players in a graphical game, and
suppose that the payoffs of $v_{2}$ and $w$ are as follows.
\begin{description}
\item [{Payoff for $v_2$:}] %
\begin{tabular}{|c||c|c|}
\hline 
 & $w$ plays $0$ & $w$ plays $1$\tabularnewline
\hline 
\hline 
$v_{2}$ plays $0$ & $0$ & $1$\tabularnewline
\hline 
$v_{2}$ plays $1$ & $1$ & $0$\tabularnewline
\hline 
\end{tabular}
\item [{Payoffs for $w$:}]~

\begin{description}
\item [{game with $v_1$:}] %
\begin{tabular}{|c||c|c|}
\hline 
 & $v_{1}$ plays $0$ & $v_{1}$plays $1$\tabularnewline
\hline 
\hline 
$w$ plays $0$ & $0$ & $\zeta$\tabularnewline
\hline 
$w$ plays $1$ & $0$ & $0$\tabularnewline
\hline 
\end{tabular}
\item [{game with $v_2$:}] %
\begin{tabular}{|c||c|c|}
\hline 
 & $v_{2}$ plays $0$ & $v_{2}$plays $1$\tabularnewline
\hline 
\hline 
$w$ plays $0$ & $0$ & $0$\tabularnewline
\hline 
$w$ plays $1$ & $0$ & $1$\tabularnewline
\hline 
\end{tabular}
\end{description}
\end{description}
Then, in every $\epsilon$-NE $\mathbf{p}\left[v_{2}\right]=\min\left(\zeta\mathbf{p}\left[v_{1}\right],1\right)\pm \epsilon$,
where $\mathbf{p}\left[u\right]$ denotes the probability that $u$
assigns to strategy $1$.
\end{lem}

\begin{proof}
(Sketch) If $\mathbf{p}\left[v_{2}\right]>\zeta\mathbf{p}\left[v_{1}\right]+\epsilon$,
then in every $\epsilon$-NE $\mathbf{p}\left[w\right]=1$, which
contradicts $\mathbf{p}\left[v_{2}\right]>\epsilon$. Similarly, if
$\mathbf{p}\left[v_{2}\right]<\min\left(\zeta\mathbf{p}\left[v_{1}\right],1\right)-\epsilon$,
then $\mathbf{p}\left[w\right]=0$, which yields a contradiction when
$\mathbf{p}\left[v_{2}\right]<1-\epsilon$.
\end{proof}

\section{Proof of Theorem \ref{thm:main}}

%\restatePCPtoBimatrix*

\subsection{Overview}
The proof is a reduction from weak approximation of Nash equilibrium in a polymatrix game (recall Theorem 2) to the two-player problem. 
The two players simultaneously play three games: the main game, which
is the heart of the reduction; and two games based on a construction
due to Althofer \cite{Alt94}, which impose structural properties
of any approximate Nash equilibrium (interestingly, the oblivious lower bound of \cite{Daskalakis-Papadimitriou-PTAS} uses the same game).  The final payoff of a player is the sum of payoffs in all three games. For convenience of notation, the payoffs in each game will be in $[0,1]$; in order to normalize the payoffs in the final game to $[0,1]$ one should multiply by $1/3$ the payoffs in all three games.

\subsubsection*{Main game}

We let each of the two players ``control'' the vertices on one side
of the bipartite graphical game (we henceforth use {\em players}
only for the players in the bimatrix game, and refer to the players
in the graphical polymatrix game as {\em vertices}). For ease of
notation, we assume wlog that each player controls an equal number
of vertices, $n$.

We partition the vertices of each player into $n/k$ disjoint subsets of
size at most $2k=2\sqrt{n}$, such that every two subsets share at most
$18$ edges.
%By Lemma \ref{lem:random-partition}, a random partition satisfies
%this property with high probability. 
By Lemma \ref{lem:derandomized-partition}, we can efficiently find such a partition.
Let $\left(S_{1},\dots,S_{n/k}\right)$
and $\left(T_{1},\dots,T_{n/k}\right)$ denote the partitions of the
respective players. Each action of the players corresponds to a choice
of a subset (out of $n/k$ subsets in the partition), and a choice
of an action for each vertex in the subset (out of  at most $2^{2k}$ vectors
of actions). Together, the main game has $\left(\frac{n}{k}\cdot 2^{2k}\right)\times\left(\frac{n}{k}\cdot 2^{2k}\right)$ action profiles. When players choose actions $\left(S_{i},\vec{\alpha}_{S_{i}}\right)$
and $\left(T_{j},\vec{\beta}_{T_{j}}\right)$, the payoff of the row player
is the sum of the payoffs over all shared edges of $S_{i}$ and $T_{j}$, when they play the respective strategies from
$\vec{\alpha}_{S_{i}}$ and $\vec{\beta}_{T_{j}}$ (here we use the polymatrix structure of the payoffs in the graphical game; the payoffs are defined over edges and therefore payoffs of a certain vertex can be defined even though not necessarily all its neighbours are playing). Finally, we normalize by $\lambda/18$, where $\lambda=\Theta(\delta^2)$ is a small constant that satisfies $\lambda>\varepsilon'$. Similarly the payoff of the second player is derived from the payoffs of the vertices in $T_{j}$. 

\ignore{
\ifFULL
\else
\begin{lem}
\label{lem:random-partition}Let $G=\left(U,V,E\right)$ be a bipartite
$d$-regular graph with $n=\left|U\right|=\left|V\right|$. Let $S_{1},\dots,S_{n/k}$
and $T_{1},\dots,T_{n/k}$ be partitions of $U$ and $V$ to disjoint
subsets of size $k=\sqrt{n}\log n$ chosen uniformly at random. Then
w.h.p.
\[
\forall i,j\in\left[n/k\right]\,\,\,\left|\left(S_{i}\times T_{j}\right)\cap E\right|<2dk^{2}/n\mbox{.}
\]
\end{lem}
\begin{proof} Using the Generalized Chernoff bound \cite{PS97-chernoff,IK10-chernoff}. See full version for details.
\end{proof}
\fi
}

\ifFULL
\else
\begin{lem}
\label{lem:derandomized-partition}
Let $G=\left(U,V,E\right)$ be a bipartite
$d$-regular graph with $n=\left|U\right|=\left|V\right|$. 
Then we can efficiently find partitions $S_{1},\dots,S_{n/k}$
and $T_{1},\dots,T_{n/k}$ of $U$ and $V$, respectively, to disjoint
subsets, such that each subset has size at most $2k=\sqrt{n}\log n$, and:
\[
\forall i,j\in\left[n/k\right]\,\,\,\left|\left(S_{i}\times T_{j}\right)\cap E\right|<2d^2k^{2}/n\mbox{.}
\]
\end{lem}
\begin{proof}
Take an arbitrary partition of $U$ into subsets $S_{1},\dots,S_{n/k}$  of size exactly $k$, 
and inductively partition $V$ into subsets $T_{1},\dots,T_{n/k}$.
Use Markov's inequality to show every vertex $v\in V$ can be placed into some subset $T_j$.
\end{proof}
\fi

\subsubsection*{Althofer's games}

Althofer's game \cite{Alt94} is an asymmetric hide-and-seek win-lose game over $l$ locations. The hider chooses a location $i\in [l]$, the seeker chooses a subset $B\subset [l]$ of of size $|B|=l/2$. The hider wins if $i\notin B$, the seeker wins if $i\in B$. Namely, the payoff function is given by 
\begin{align*}
u_1(i,B)=1-u_2(i,B)=\begin{cases} 
1 & \text{if } i\notin B \\
0 & \text{if } i\in B.
\end{cases}
\end{align*}
In Althofer's game each player can guarantee $\frac{1}{2}$ by uniform play, therefore the value of the game is $\frac{1}{2}$. Althofer's game enjoys the following strong property: in every $\varepsilon$-approximate Nash equilibrium, the hider must play a mixed strategy that is $O(\varepsilon)$-close to the uniform distribution in total variation distance; see \cite{Daskalakis-Papadimitriou-PTAS}.

In our game each player plays two (unrelated) Althofer's games, one game as a hider, and one as a seeker. When player 1 is a hider, the locations are the sets $S_{1},\dots,S_{n/k}$ (i.e., $l=n/k$). The action from the main game $\left(S_{i},\vec{\alpha}_{S_{i}}\right)$ determines the location, where for the purposes of Althofer's game we ignore $\vec{\alpha}_{S_{i}}$. When player 1 is a seeker, the locations are the sets $T_{1},\dots,T_{n/k}$, and player 1 chooses an action to play in this game, independent of $\left(S_{i},\vec{\alpha}_{S_{i}}\right)$. The total number of actions for player 1 is $\left(\frac{n}{k}\cdot 2^{2k}\right)\cdot{n/k \choose n/\left(2k\right)}=2^{\tilde{O}(\sqrt{n})}$. Similarly, for player 2, who is playing as a hider over $T_{1},\dots,T_{n/k}$, and as a seeker over $S_{1},\dots,S_{n/k}$.

\subsection{Structure of an equilibrium}
For a mixed action $x$ of player $1$ we denote by $x(S_i)$ the total probability that the player chooses to control the vertices $S_i$; i.e., 
\begin{align*}
x(S_i)=\underset{\vec{\alpha}\in 2^{[k]}, B\subset \{T_1,...,T_{n/k}\}, |B|=\frac{n}{2k}}{\sum} x(((S_i,\vec{\alpha}),B)).
\end{align*}

\begin{lem}\label{lem:uniform}

If $(x,y)$ is an $\epsilon'$-Nash equilibrium,
then $\sum_{i\in [n/k]}\left|x\left(S_{i}\right)-\frac{k}{n}\right|=O\left(\lambda\right)$.
\end{lem}

In the proof we use the following lemma from the full version of \cite{Daskalakis-Papadimitriou-PTAS}:
\begin{lem}[\cite{Daskalakis-Papadimitriou-PTAS}]\label{lem:DP-PTAS}
Let $\{a_i\}_{i\in [l]}$ be real numbers satisfying the following properties for some $\theta>0$:
%\begin{enumerate}
(1) $a_1\geq a_2 \geq ... \geq a_l$;
(2) $\sum_{i\in [l]} a_i = 0$; and
(3) $\sum_{i\in [l/2]}a_i \leq \theta$.
%\end{enumerate}
Then $\sum_{i\in [l]} |a_i|\leq 4\theta$.
\end{lem}
\begin{proof}[Proof of Lemma \ref{lem:uniform}]
In order apply Lemma \ref{lem:DP-PTAS}, we denote $a_i=x(S_i)-\frac{k}{n}$ and we assume wlog that $x(S_1)\geq x(S_2)\geq ... \geq x(S_{n/k})$. Then the first two conditions hold. Regarding the third condition, we argue that $\sum_{i\in [l/2]} x(S_i)\leq 3\lambda$ this will complete the proof.

Player 1 can guarantee a payoff of $1/2$ in the sum of Althofer's games as a hider and the main game, by playing uniformly and choosing an arbitrary actions for the controlled vertices (for instance, $\vec{\alpha}_{S_i}=\vec{0}$). Assume by contradiction that $\sum_{i\in [l/2]} x(S_i)> 3\lambda$. Since Player 2 is $\epsilon'$-best replying, his payoff in the Althofer game (as a seeker) is at least $\frac{1}{2}+3\lambda-\epsilon'$ (because he can get $1/2+3\lambda$ by choosing the set $[l/2]$). Therefore, Player 1's payoff in the Althofer game (as a hider) is at most $\frac{1}{2}-3\lambda+\epsilon'$. If we add to it his payoffs in the main game, then his payoff is at most $\frac{1}{2}-2\lambda+\epsilon'$. Therefore, player 1 can increase his payoff by $2\lambda-\epsilon'>\epsilon'$ by deviating to the uniform distribution over the locations (as a hider), and maintaining his mixed action as a seeker.
\end{proof}

\subsection{Completing the proof}

Any mixed strategy $x$ of player 1 in the bimatrix game induces a mixed strategy of all the vertices in $\cup_i S_i$ in the obvious way. Vertex $s\in S_i$ plays the action $1$ with probability $p$ where $p$ is the conditional probability
\begin{align*}
Pr_x(\vec{\alpha}_{S_i}(s)=1|\text{Player 1 controls } S_i).
\end{align*}
If the event ``Player 1 controls $S_i$" occurs with probability 0, then we define (arbitrarily) that $p=1$. Similarly, any mixed strategy $y$ of player 2 induces a mixed strategy of all the vertices in $\cup_i T_i$.

We claim that if $\left(x,y\right)$ is an $\epsilon'$-approximate
Nash equilibrium of the bimatrix game, then the induced mixed-strategies
profile is an $\left(\epsilon,\delta\right)$-approximate Nash equilibrium
of the original graphical game. 

By Lemma \ref{lem:uniform} and Markov's inequality, for a $\left(1-O\left(\sqrt{\lambda}\right)\right)$-fraction
of the subsets, the row player distributes within a $\left(1\pm O\left(\sqrt{\lambda}\right)\right)$-factor
of the correct weight ($k/n$); i.e. $x\left(S_{i}\right)\in\frac{k}{n}\cdot\left[1-O\left(\sqrt{\lambda}\right),1+O\left(\sqrt{\lambda}\right)\right]$.
Let us restrict our attention only to those subsets, and call them
{\em good}. We say that a vertex is {\em good} if it and all of
its neighbors belong to good subsets. Since the game graph is of bounded
degree, we again have that a $\left(1-O\left(\sqrt{\lambda}\right)\right)$-fraction of the vertices is good. 

Consider any good vertex whose induced mixed strategy is not $\epsilon$-optimal
in the original game given that the rest of the vertices also play
according to their induced strategies. Then, changing its strategy
in the bimatrix game (while leaving all other marginals the same),
would increase the payoff of its player by at least
\[
\left(1-O\left(\sqrt{\lambda}\right)\right)^2\cdot\left(\frac{k}{n}\right)^{2}\cdot \left(1-O\left(\sqrt{\lambda}\right)\right)\cdot \epsilon \cdot \frac{\lambda}{18}=\Omega\left(\epsilon\lambda/n\right)\mbox{,}
\]
where the $\left(1-O\left(\sqrt{\lambda}\right)\right)^2\cdot\left(\frac{k}{n}\right)^{2}$ term corresponds to the probabilities that the subsets corresponding to both the vertex and any of its neighbors are played; the $\left(1-O\left(\sqrt{\lambda}\right)\right)\cdot \epsilon$ term corresponds to the improvement of the vertex in the bimatrix game, where instead of summing the payoff in all three edges (as in the graphical game) these three edges have weights $\{1+\gamma_j\}_{j=1,2,3}$ for $|\gamma_j|=O(\sqrt{\lambda})$; finally $\lambda/18$ is the normalization. The right side follows by plugging in $k=\sqrt{n}$.

If (by a contradiction) a $\delta$-fraction of the vertices have $\epsilon$-improvement strategies; then one of the players has $\left(\delta/2-O\left(\sqrt{\lambda}\right)\right)\cdot n$ good vertices with an $\epsilon$-improvement\footnote{Here we use our choice of $\lambda=\Theta(\delta^2)$, which guarantees that $\delta/2-O\left(\sqrt{\lambda}\right)=\Omega(\delta)$}. This player can benefit $\Omega\left(\epsilon\cdot\lambda/n\right)$ from a deviation of each vertex. So his total improvement from all deviations simultaneously is  $\Omega\left(\epsilon\cdot\delta\cdot\lambda\right)$ - which is impossible when $(x,y)$ is an $\epsilon'$-approximate Nash equilibrium for $\epsilon'=\Theta(\epsilon\cdot\delta\cdot\lambda)$.

\section{Implications for Relative $\epsilon$-approximate Nash}

Daskalakis \cite{Das13_multiplicative_hardness} defines a notion of relative (sometimes also called {\em multiplicative} \cite{HRS08-multiplicative_Nash}, as opposed to the more standard additive) $\epsilon$-Nash equilibrium, and proves that in two-player games with payoffs in $\left[-1,1\right]$, finding a relative $\epsilon$-well supported Nash equilibrium is \PPAD-complete. In particular, he concludes that the quasi-polynomial algorithm of Lipton et al. \cite{LMM03_quasi_poly} cannot achieve this notion of approximate equilibrium.

This result has two caveats: (1) Through the use of positive and negative payoffs, the gain from deviation is large compared to the expected payoff, only because the latter is small due to cancellation of positive and negative payoffs. Namely, the gain from deviation may be very small compared to the expected magnitude of the payoff. (2) It only applies to the more restrictive notion (thus rendering the hardness result weaker) of {\em well-supported} approximate equilibrium; i.e. an equilibrium where every action in the support has to be approximately optimal. Showing \PPAD-hardness for both positive payoffs and general (non well-supported) approximate equilibrium were left as open questions in 
\cite{Das13_multiplicative_hardness}. Recently, the first question was settled in \cite{Rub14-polymatrix-Nash} where it was shown that finding a relative $\epsilon$-well-supported Nash equilibrium with positive payoffs is indeed \PPAD-complete. 

Here, assuming Conjecture 2, we settle the second question: we show that finding any relative $\epsilon$-approximate Nash equilibrium is \PPAD-complete. Furthermore, in our hard instance the row player has only positive payoffs and the column player has only negative payoffs, and so there is no cancellation of payoffs as in the construction of \cite{Das13_multiplicative_hardness}.

\begin{thm}
\label{thm:relative} Assuming Conjecture \ref{conj:pcp} there exists
a constant $\epsilon'=\Theta\left(\epsilon\cdot\delta^{3}\right)>0$
such that finding a relative $\epsilon'$-approximate Nash equilibrium
(ANE) in a bimatrix game where the row player's payoffs are non-negative
and the column player's payoffs are non-positive is \PPAD-complete.\end{thm}

Proving a similar theorem when both players have positive payoffs remains an interesting open question.
%(Note that Theorem \ref{thm:main} applies to this setting, yet it can give at most quasi-polynomial hardness.)
In fact, we do not know of any instances with positive payoffs where all $\epsilon$-approximate Nash equilibria must have large (e.g., linear, or even super-logarithmic) support.

\begin{proof}
We reduce from the problem of finding an (additive) $\left(\epsilon,\delta\right)$-Nash
in a bipartite, degree $3$, polymatrix game with two actions per
player. We construct a main game where the bimatrix game players (henceforth
just {\em players}) control the nodes of the polymatrix game, and
two side games that guarantee that each player randomizes (approximately)
uniformly over all her nodes.
\begin{description}
\item [Main game] We let the row player ``control'' the nodes on
one side of the bipartite game graph, and let the column player control
the nodes on the other side. Namely, let $n$ be the number of nodes
on each side of the graph; each player has $2n$ actions, each corresponding
to a choice of node and an action for that node. When the players
play strategies that correspond to adjacent nodes in the graphs, they
receive the payoffs of the corresponding nodes, scaled (by a small
positive constant $\eta=O\left(\delta^{2}\right)$) and shifted to
fit in the intervals: $\left[1,1+\eta\right]$ for the row player,
and $\left[-\left(1+\eta\right),-1\right]$ for the column player.
If the nodes played do not share an edge in the bipartite game graph,
the utility for both players is zero. Notice that if either player
chooses her node uniformly at random (and chooses the action for that
node arbitrarily), then the expected payoff is in $\frac{3}{n}\cdot\left[1,1+\eta\right]$
for the row player, and in $\frac{3}{n}\cdot\left[-\left(1+\eta\right),-1\right]$
for the column player. 
\item [Side games] We also let the players play two hide-and-seek
win-lose zero-sum games over a space of $n$ actions. In both games,
the row player is chasing the column player. In each game, if they
pick the same strategy, the row player receives payoff $1$ (and the
column player receives $-1$); otherwise the payoffs are $0$. Finally,
in the first side game we identify the row player's strategies with
her choice of nodes in the main game. Namely, if she plays node $i$
in the main game and the column player chose $i$ in the first side
game, then her payoff from this side game is $1$. Similarly, we identify
the column player's strategies in the second game with her choice
of nodes in the main game. 
\end{description}
We proceed by showing that in every relative $\epsilon'$-ANE, the
row player's utility is approximately $5/n$, and the column player's
utility is approximately $-5/n$ (Lemma \ref{lem:app-zero_sum});
then we show that in every relative $\epsilon'$-ANE, both players
randomize approximately uniformly over their nodes (Lemma \ref{lem:structure});
finally we use these two observations to complete the proof (Subsection
\ref{sub:Completing-the-proof}).
\end{proof}

\subsection{Value and structure of a relative $\epsilon'$-ANE}

Given mixed strategies $\left(x,y\right)$, we let $x\left(i\right)$
denote the total probability that the row player assigns to node $i$,
and analogously for $y\left(j\right)$. We also let $x^{*}\left(y\right)$
and $y^{*}\left(x\right)$ denote the corresponding best responses
of each player. Finally, let $U_{R}\left(x;y\right)$ and $U_{C}\left(y;x\right)$
denote the expected payoffs for the row and column players, respectively.
\begin{lem}
\label{lem:app-zero_sum}If $\left(x,y\right)$ is a relative $\epsilon'$-ANE,
then 
\begin{equation}
\left(1-\epsilon'\right)\cdot\frac{5}{n}\leq U_{R}\left(x;y\right)\leq\left(1+\eta\right)\left(1+\epsilon'\right)\cdot\frac{5+3\eta}{n}\label{eq:UR-bounds}
\end{equation}
and 
\begin{equation}
-\left(1-\epsilon'\right)\left(1-\mathbf{\eta}\right)\cdot\frac{5}{n}\geq U_{C}\left(y;x\right)\geq-\left(1+\epsilon'\right)\cdot\frac{5+3\eta}{n}\label{eq:UC-bounds}
\end{equation}
\end{lem}
\begin{proof}
Observe that the main game is relative $\eta$-approximately zero-sum;
i.e. for any pure strategy profile $\left(x',y'\right)$ the expected
utilities $U_{R}^{\mbox{main}}\left(x';y'\right),U_{C}^{\mbox{main}}\left(y';x'\right)$
of the row and column player, respectively, satisfy:
\begin{equation}
\left|U_{R}^{\mbox{main}}\left(x';y'\right)+U_{C}^{\mbox{main}}\left(y';x'\right)\right| 
	\leq\eta\cdot\min\left\{ U_{R}^{\mbox{main}}\left(x';y'\right),-U_{C}^{\mbox{main}}\left(y';x'\right)\right\} \label{eq:eps-zero_sum}
\end{equation}
By linearity of expectation and convexity of the absolute value function, this continues to hold when $\left(x',y'\right)$
are mixed strategies. Furthermore, since the side games are exactly
zero-sum and have the same signs as the main game, the same follows
for the payoffs from the entire game:

\begin{equation}
\left|U_{R}\left(x';y'\right)+U_{C}\left(y';x'\right)\right|\leq\eta\cdot\min\left\{ U_{R}\left(x';y'\right), -U_{C}\left(y';x'\right)\right\} \label{eq:eps-zero_sum-1}
\end{equation}
In particular, the above inequality holds for $\left(x,y\right)$.
Thus, the upper bounds in (\ref{eq:UR-bounds}) and (\ref{eq:UC-bounds})
follow from the lower bounds in (\ref{eq:UC-bounds}) and (\ref{eq:UR-bounds}),
respectively.

Finally, $\left(1-\epsilon'\right)\cdot\frac{5}{n}\leq U_{R}\left(x;y\right)$
in every relative $\epsilon'$-ANE, because the row player can guarantee
an expected payoff of at least $\frac{5}{n}$ by randomizing uniformly
over all her strategies. Similarly, $U_{C}\left(y;x\right)\geq-\left(1+\epsilon'\right)\left(\frac{3\left(1+\eta\right)+2}{n}\right)$
because the column player can guarantee a payoff of at least $-\frac{3\left(1+\eta\right)+2}{n}$
by randomizing uniformly over all her strategies.\end{proof}
\begin{lem}
\label{lem:structure}There exists a constant $\lambda=\Theta\left(\delta^{2}\right)$
such that for every $\left(x,y\right)$ relative $\epsilon'$-ANE,
$\sum\left|x\left(i\right)-\frac{1}{n}\right|\leq\lambda$ and $\sum\left|y\left(j\right)-\frac{1}{n}\right|\leq\lambda$.
\begin{proof}
Assume by contradiction that $\sum\left|x\left(i\right)-\frac{1}{n}\right|>\lambda$,
then in particular there must exist an $i\in\left[n\right]$ such
that $x\left(i\right)<\left(1-\lambda/2\right)/n$. When the column
player chooses her strategy uniformly at random in the main game and in the second
side game and plays strategy $i$ in the first side game, her expected
payoff is at least 
\[
U_{C}\left(y^{*}\left(x\right);x\right)\geq-\frac{3}{n}\left(1+\eta\right)-\frac{1}{n}-x\left(i\right).
\]
Therefore, in any relative $\epsilon'$-ANE her expected utility is
at least 
\[
U_{C}\left(y;x\right)\geq\left(1+\epsilon'\right)\cdot\left(\frac{-5-3\eta+\lambda/2}{n}\right),
\]
Which contradicts the upper bound in (\ref{eq:UC-bounds}) when we
take $\lambda$ sufficiently large, e.g. $\lambda=10\eta$.

Similarly, if $\sum\left|y\left(j\right)-\frac{1}{n}\right|>\lambda$,
there must exist an $j\in\left[n\right]$ such that $y\left(j\right)>\left(1+\lambda/2\right)/n$.
Therefore, the row player can guarantee a payoff of at least 
\[
U_{R}\left(x^{*}\left(y\right);y\right)\geq\frac{4}{n}+y\left(j\right).
\]
Thus by relative $\epsilon'$-ANE, 
\[
U_{R}\left(x;y\right)\geq\left(1-\epsilon'\right)\cdot\left(\frac{5+\lambda/2}{n}\right),
\]
which contradicts the upper bound in (\ref{eq:UR-bounds}).
\end{proof}
\end{lem}

\subsection{\label{sub:Completing-the-proof}Completing the proof of Theorem
\ref{thm:relative}}

Now, given a relative $\epsilon'$-ANE $\left(x,y\right)$, we can
take, for each node $i$, the mixed strategy induced by the probabilities
$x\left(i:1\right)/x\left(i\right)$ and $x\left(i:2\right)/x\left(i\right)$
that the row player assigns to each action (respectively, $y\left(j:1\right)/y\left(j\right)$
and $y\left(j:2\right)/y\left(j\right)$ assigned by the column player).
We claim that this strategy profile is an $\left(\epsilon,\delta\right)$-approximate
equilibrium for the polymatrix game. Assume by contradiction that
this is not the case.

By Lemma \ref{lem:structure} and Markov's inequality, a $\left(1-O\left(\sqrt{\lambda}\right)\right)$-fraction
of the nodes are played within $\frac{\sqrt{\lambda}}{n}$ of the
correct probability $1/n$. For a node $i$ controlled by the row
player, we say that it is {\em good} if $\left|x\left(i\right)-1/n\right|\leq\frac{\sqrt{\lambda}}{n}$
and $\left|y\left(j\right)-1/n\right|\leq\frac{\sqrt{\lambda}}{n}\;\;\forall j\in{\cal N}\left(i\right)$,
and analogously for column player's nodes. Since the graph has bounded
degree, a $\left(1-O\left(\sqrt{\lambda}\right)\right)$-fraction
of the nodes are good. 

Let $i$ be any good node who has an $\epsilon$-improving deviation
from her induced strategy in the polymatrix game. If the player who
controls $i$ makes the corresponding deviation in the two-player
game, she increases her expected payoff by at least $\epsilon\cdot\eta\cdot\frac{1}{n^{2}}\cdot\left(1-O\left(\sqrt{\lambda}\right)\right)$.
(We multiply by $\eta$ to account for the scaling; by $1/n^{2}$
for the probability that this player plays node $i$ and the other
player plays a neighbor of $i$; and by $\left(1-O\left(\sqrt{\lambda}\right)\right)$
to correct for the deviation from $1/n$ in the latter probabilities.)
By our assumption that the induced strategy profile is not an $\left(\epsilon,\delta\right)$-approximate
equilibrium for the polymatrix game, at least one of the players has
at least $\left(\delta-O\left(\sqrt{\lambda}\right)\right)n$ good
nodes with $\epsilon$-improving deviations. Summing her gains from
those deviations, we get that this player has can improve her expected
payoff by at least $\left[\epsilon\cdot\eta\cdot\left(1-O\left(\sqrt{\lambda}\right)\right)\right]\left(\delta-O\left(\sqrt{\lambda}\right)\right)/n$.

However, this is a contradiction since it follows from Lemma \ref{lem:app-zero_sum},
that $\left(x,y\right)$ is also an (additive) $\epsilon'\cdot\left(\frac{6}{n}\right)$-ANE.
\qed

\section{Implications for Fairness Mechanisms}
Competitive equilibrium from equal incomes (CEEI) is a well-known
fair allocation mechanism \cite{Foley67:Resource,Varian74:Equity,Thomson85:Theories}: We give all agents a unit of money, and price the goods in such a way that the market clears.  It is also well-known that when goods are indivisible or utilities
are non-linear, an equilibrium may not exist.  However, Budish \cite{ACEEI_Bud11}
proves that an approximate CEEI still exists. This
concept of approximate equilibrium is used in practical system for allocating seats in courses to business school students \cite{wharton_paper}.

Budish \cite{ACEEI_Bud11} measures the proximity to a perfect CEEI
via two parameters: a solution is an $\left(\alpha,\beta\right)$-CEEI
if the clearing error of the competitive equilibrium is less than $\alpha$, and all the incomes
are between $1$ and $1+\beta$. Budish shows that an $\left(\alpha,\beta\right)$
always exists for some favorable $\alpha=\alpha^{*}$, and any $\beta>0$.
Recently, \cite{OPR14_A-CEEI} showed a reduction from $\epsilon$-{\sc Gcircuit} with fan-out $2$,
to the problem of finding an $\left(\alpha^{*},\Theta\left(\epsilon /\log\left(1/\epsilon\right)\right)\right)$-CEEI.
In particular, when combined with the results of \cite{Rub14-polymatrix-Nash},
this implies that it is \PPAD-complete to find an $\left(\alpha^{*},\beta\right)$-CEEI
for some constant $\beta>0$. 

\ifFULL
The $\beta$ parameter used in Budish's formulation is an imperfect way of measuring income inequalities, as it may be set by a single outlier.  Perhaps the best known, and most widely used, measure of income inequality is the {\em Gini index\/} (e.g. \cite{gini1912variabilita,cowell2011measuring}) (see the Appendix for the definition).  In fact, the Gini index is used to assess the performance of the class seat assignment system used in practice \cite{wharton_paper}.  
\else
The $\beta$ parameter used in Budish's formulation is an imperfect way of measuring income inequalities, as it may be set by a single outlier.  Perhaps the best known, and most widely used, measure of income inequality is the {\em Gini index\/} \cite{gini1912variabilita,cowell2011measuring} (see e.g. \cite{cowell2011measuring} or our full version for definition).  In fact, the Gini index is used to assess the performance of the class seat assignment system used in practice \cite{wharton_paper}.  
See the full version for the precise definitions and statement of the result.
\fi

\begin{thm}[Informal] \label{thm:ceei-informal}
Assuming Conjecture \ref{conj:pcp},
finding an income assignment and prices with low market clearing error and near-optimal Gini index is PPAD-complete.
\end{thm}
%
%Before we formally state this result, 
%we should define the important terms from above paragraphs.

\ifFULL
See the appendix for the precise definitions and statement of the result.
\fi

\section{Discussion}
The purpose of this paper is to showcase an important open problem:
\begin{quote}
{\bf reduce {\sc End of the Line} to $(\epsilon,\delta)$-{\sc Gcircuit}}
\end{quote}
--- i.e., show that Conjecture \ref{conj:eth} implies Conjecture \ref{conj:pcp}.
As we mentioned, such a reduction would imply a ``PCP for \PPAD"; 
i.e. it would imply that there exists a probabilistically checkable proof for a \PPAD-complete problem.

The equivalent statement for the class of \NP-complete problem, 
is the celebrated PCP Theorem \cite{AS98-PCP, ALMSS98-PCP}.
What can we learn from our experience with constructing PCP's for \NP?
First of all, we have many different constructions of PCPs for \NP, 
and this may be seen as circumstantial evidence that constructing a PCP for \PPAD~is possible.
Equally important, there are several constructions of PCPs for \NP~that are near-linear 
(e.g. \cite{bSGHSV06-quasilienar_PCP, bSS06-quasilinear_PCP, Dinur07-PCP, MR10-quasilinear_PCP});
notice that for the application to $\epsilon$-Nash in bimatrix games we need the proof length to be sub-quadratic.

The next question one should ask
is, whether we can adapt the techniques used in the proofs of the PCP Theorem 
to the class of \PPAD-complete problems.
We briefly and informally sketch some of our thoughts on the matter.
All the proofs of the PCP Theorems that we are aware of (including \cite{AS98-PCP, ALMSS98-PCP, Dinur07-PCP})
compose an {\em inner verifier} with an {\em outer verifier}. 
The outer verifier in Dinur's proof \cite{Dinur07-PCP} is combinatorial in nature,
and it seems plausible that the same or similar techniques could be modified to fit the generalized circuit graph.
Unfortunately, all inner verifiers that we know are discrete in nature,
whereas our characterization of \PPAD~with {\sc Gcircuit} (or Nash)
is inherently based on continuous constraints.
Thus the following interesting (and somewhat open-ended) questions arise:
Can we find a characterization of \PPAD~via a constraint satisfaction problem  
whose constraints are discrete in nature?
Alternatively, can one construct an inner (non-efficient) verifier 
using the constraints specified in the definition of the {\sc Gcircuit} problem?

\paragraph{Acknowledgments:}  Many thanks to Boaz Barak, Paul Cristiano, Muli Safra, and Madhu Sudan for inspiring discussions about probabilistic checkable proofs. 
Special thanks to Mark Braverman for pointing out a subtlety in the original proof of Theorem 2 in the conference proceedings.

\bibliographystyle{alpha}
\bibliography{PCP}

\appendix

\section{Finding a good partition}

\ifFULL
	\begin{lem}
	\label{lem:derandomized-partition}
\else
	\begin{lem*}
	[Lemma \ref{lem:derandomized-partition}]
\fi
Let $G=\left(U,V,E\right)$ be a bipartite
$d$-regular graph with $n=\left|U\right|=\left|V\right|$. 
Then we can efficiently find partitions $S_{1},\dots,S_{n/k}$
and $T_{1},\dots,T_{n/k}$ of $U$ and $V$, respectively, to disjoint
subsets, such that each subset has size at most $2k$, and:
\[
\forall i,j\in\left[n/k\right]\,\,\,\left|\left(S_{i}\times T_{j}\right)\cap E\right|\leq 2d^2k^{2}/n\mbox{.}
\]
\ifFULL
	\end{lem}
\else
	\end{lem*}
\fi
\begin{proof}
Let $S_{1},\dots,S_{n/k}$ be an arbitrary partition of $U$ into disjoint subsets of size exactly $k$.
We inductively place the vertices of $V$ into subsets $T_{1},\dots,T_{n/k}$,
while maintaining the desiderata that each subset $T_j$ is of size at most $2k$,
and for every $i,j$ the number of edges from $S_i$ to $T_j$ is at most $2d^2k^2/n$.

It is left to show that for any partial partition of $V$, there is a subset $T_J$ into which we can place the next vertex $v$.
In expectation, every subset $T_j$ has less than $k$ vertices. 
Therefore by Markov's inequality, less than half of the subsets have $2k$ vertices or more.
$v$ has neighbors in at most $d$ subsets $S_i$. 
Recall that the $S_i$'s are of size exactly $k$. 
Thus each $S_i$ with a neighbor of $v$ has, in expectation over $j$, less than $dk^2/n$ neighbors in $T_j$.
Using Markov's inequality again, less than a $1/(2d)$-fraction of the subsets $T_j$ contain $2d^2k^2/n$ neighbors of $S_i$. 
In total, we lose less than half the subsets for the size desideratum, 
and less than $1/(2d)$-fraction of the subsets for each $S_i$ containing a neighbor of $v$.
Therefore there always remains at least one subset $T_j$ to which we can add $v$. 
\end{proof}

\section{The course allocation problem}
Even though the approximate CEEI and the existence theorem in \cite{ACEEI_Bud11} are applicable to a broad range of allocation problems, we shall describe our results in the language of the course allocation problem.

We are given a set of $M$ courses with integer capacities (the supply) $(q_{j})_{j=1}^{M}$, and a set of $N$ students, where each student $i$ has a set $\Psi_{i}\subseteq 2^{M}$
of permissible course bundles, with each bundle containing at most $k \leq M$ courses.  The set $\Psi_{i}$
encodes both scheduling constraints (e.g., courses that meet
at the same time) and any constraints specific to student $i$ (e.g. prerequisites).

Each student $i$ has a strict ordering over her permissible schedules, denoted by $\preccurlyeq_{i}$.  We allow arbitrarily complex preferences --- in particular,
students may regard courses as substitutes or complements.  More formally:

\begin{defn} \textbf{Course Allocation Problem}
The input to a course allocation problem consists of:
\begin{itemize}
\item For each student $i$ a set of course bundles $\left(\Psi_{i}\right)_{i=1}^{N}$.
\item The students' reported preferences, $\left(\preccurlyeq_{i}\right)_{i=1}^{N}$,
\item The course capacities, $\left( q_j \right)_{j=1}^{M}$, and
\end{itemize}

The output to a course allocation problem consists of:
\begin{itemize}
\item Prices for each course $(p_{j}^{*})_{j=1}^{M}$,
\item Allocations for each student$(x_{i}^{*})_{i=1}^{N}$, and
\item Budgets for each student $\left(b_{i}^{*}\right)_{i=1}^{N}$.
\end{itemize}

\end{defn}

The quality of an allocation is evaluated based on its proximity to market clearing and income equality.
As for the definition of market clearing error,
it suffices for our purposes to require that no course is over-subscribed or under-subscribed by more than $\alpha^*$ students
(we refer the curious reader to \cite{ACEEI_Bud11} for the precise definition).
The Gini coefficient, which measures income inequality, is discussed in the following subsection.

\subsection{The Gini coefficient}

\begin{defn}[Gini index]
Given a distribution of incomes $D$,
the {\em Lorenz curve} plots, for each $x \in [0,1]$, the cumulative wealth
owned by the bottom $x$-fraction of the population. 
Let $F_D^{-1}(x) = \sup \big\{y \colon \Pr_{z\sim D}[z \leq y] \leq x \big\}$,
and define $L_D(x) = \left( \int_0^x F_D^{-1}(x) \right) / \left( \int_0^1 F_D^{-1}(x) \right)$.
Then the Lorenz curve is the graph $\big(x,L_D(x)\big)$, for $x \in [0,1]$.
Notice that $F_D^{-1}(x)$ is monotonically non-decreasing, so the Lorenz curve is convex.

The {\em Gini index} is defined as the ratio of the area between the $\big(x,x\big)$ line
and the Lorenz curve (by the convexity, $x \geq L_D(x) \forall x \in [0,1]$),
divided by the entire area under the $\big(x,x\big)$ line (the latter is always $1/2$):
$G_D = 2 \int_0^1 (x - L_D(x)) dx = 1 - 2 \int_0^1 L_D(x) dx$
\end{defn}

In general, a smaller Gini index corresponds to a more equal distribution of wealth.
For example, when all incomes are exactly equal, 
the Lorenz curve is exactly equal to the $\big(x,x\big)$ line,
and the Gini index is $0$.
On the other extreme, when one person has all the wealth,
the area under the Lorenz curve goes to $0$ as the population size increases,
and the Gini index approaches $1$.

\subsection{Intractability of approximate CEEI with near-optimal Gini index}

\begin{thm*} [Theorem \ref{thm:ceei-informal}, formal]
%\label{cor:CEEI}
Assuming Conjecture \ref{conj:pcp}, there exists some constant
$\gamma>0$ such that finding an allocation with market clearing error
$\alpha^{*}$ and Gini coefficient $\gamma$ is \PPAD-complete.
\end{thm*}

The rest of this section is devoted to sketching a proof of Theorem \ref{thm:ceei-informal}.
In the next subsection, we briefly outline the reduction of \cite{OPR14_A-CEEI} 
from generalized circuits with fan-out $2$ to approximate CEEI.
(Recall that in Section \ref{sec:weakNash} we show how to convert any generalized circuit to one with fan-out $2$.)
Then, we show that after normalizing the median budget to $1+\epsilon'/2$, 
in any allocation which does not correspond 
to a valid solution to the $(\epsilon,\delta)$-{\sc Gcircuit} instance, 
a $\delta'$-fraction of the students have budgets either 
at most $1$ or at least $1+\epsilon'$, 
(for some constants $\epsilon'=\Theta\left(\epsilon / \log(1/\epsilon)\right)$ and 
$\delta'=\Theta\left(\delta / \log(1/\epsilon)\right)$).
Then, the proof is complete with the following lemma:

\begin{lem}
\label{lem:gini}
Let the median income be $1+\epsilon'/2$, 
and suppose that a $\delta'$-fraction of the population 
has income at most $1$ (resp. at least $1+\epsilon'$).
Then the Gini index is at least $\gamma = \Theta(\delta'\cdot \epsilon') = \Theta(\delta\cdot \epsilon /\log^2(1/\epsilon))$.
\end{lem}

\begin{proof}
The total income of the poorer half of the population is at most $(1/2-\delta')(1+\epsilon'/2) + \delta' =  (1/2)(1+\epsilon'/2) - \delta'\cdot\epsilon'/2$,
whereas the richer half of the population has a total income of at least $(1/2)(1+\epsilon'/2)$. 
Therefore, the Lorenz curve's value at $1/2$ is bounded by 
\begin{gather*}
L(1/2) \leq \frac{(1/2)(1+\epsilon'/2) - \delta'\cdot\epsilon'/2} {(1+\epsilon'/2) - \delta'\cdot\epsilon'/2} = 1/2-\Theta(\delta' \cdot \epsilon').
\end{gather*}
We can now use elementary geometry to bound the area under the Lorenz curve:
\begin{align*}
 \int_0^1 L(x)dx &=\int_0^{1/2} L(x)dx + \int_{1/2}^1 L(x)dx\\
& \leq \frac{1}{4} \Big[0 + \big(1/2-\Theta(\delta' \cdot \epsilon')\big)\Big] 
+ \frac{1}{4} \Big[ \big(1/2-\Theta(\delta' \cdot \epsilon')\big) + 1 \Big]\\
& = 1/2 - \Theta(\delta' \cdot \epsilon').
\end{align*}
Therefore the Gini index is at least $\Theta(\delta' \cdot \epsilon')$. A similar argument works when a $\delta'$-fraction of the population has income at least $1+\epsilon'$.

\end{proof}

\subsection{From generalized circuits to Course Allocation}

\cite{OPR14_A-CEEI} reduce generalized circuits to Course Allocation,
by constructing a gadget for each gate of the generalized circuit. 
(In fact, a few gadgets per gate are required to handle a subtle issue that \cite{OPR14_A-CEEI} call "course-size amplification".)
We provide the gadget for the NOT gate as an example,
and then describe the properties of the \cite{OPR14_A-CEEI}'s reduction that we need to complete the proof.

We henceforth normalize the prices and budgets in every assignment to the Course Allocation instance 
such that the median budget is $1+\epsilon'/2$.

\begin{lem} [Essentially \cite{OPR14_A-CEEI}]
Let $n_{x}>6\alpha^*$ and suppose that the economy contains the following
courses:
\begin{itemize}
\item $c_{x}$ (the ``input course'') ;
\item $c_{{1-x}}$ with capacity $q_{{1-x}}=2n_{x}/3$ (the
``output course'');
\end{itemize}

and the following set of students:
\begin{itemize}
\item \begin{flushleft}
$n_{x}$ students interested only in the schedule $\left\{ c_{x},c_{{1-x}}\right\} $;
\par\end{flushleft}
\end{itemize}

and suppose further that at most $n_{{1-x}}=n_{x}/6$ other
students are interested in course $c_{{1-x}}$.

Then in any normalized approximate CEEI with market clearing error at most $\alpha^*$, 
at least one of the following must hold:
\begin{description}
\item[The gate is $\epsilon'$-satisfied]
\[
p_{{1-x}}^{*}\in\left[1-p_{x}^{*},1-p_{x}^{*}+\epsilon'\right]
\]
\item[The gadget contributes to income inequality]
A constant fraction of the $n_x$ students have budgets either less than $1$ or greater than $1+\epsilon'$.
\end{description}

\end{lem}
\begin{proof}
Observe that:
\begin{itemize}
\item If $p_{{1-x}}^{*}>1-p_{x}^{*}+\epsilon'$, then none of the $n_{x}$
students can afford the bundle $\left\{ c_{x},c_{{1-x}}\right\}$
 - except those whose budget is greater than $1+\epsilon'$.
Other than them, there are at most $n_{{1-x}}=n_{x}/6$ students
enrolled in the $c_{{1-x}}$ - much less than the capacity
$2n_{x}/3$. Therefore for the market clearing error to be less than $\alpha^* = n_x/6$,
$n_x/3$ of the students must have budget greater than $1+\epsilon'$
\item On the other hand, if $p_{{1-x}}^{*}<1-p_{x}^{*}$, then all
$n_{x}$ students can afford the bundle $\left\{ c_{x},c_{{1-x}}\right\}$
 - except for those whose budget is less than $1$.
Therefore in order to satisfy the market clearing requirement, at least $n_x/6$ students must have budget less than $1$.
\end{itemize}
\end{proof}

Similarly, \cite{OPR14_A-CEEI} provide gadgets for all the gates in the definition of the {\sc Gcircuit} problem,
and show that they can be concatenated to simulate the computation on the generalized circuit.
For each gate of the generalized circuit, \cite{OPR14_A-CEEI}'s reduction uses at most $\Theta(1/\log(1/\epsilon))$
(in particular, a constant) number of gadgets.
Furthermore, the number of students that participate in each of those gadgets is approximately the same ($\Theta(\alpha^*)$).
Therefore, for every gate which is not $\epsilon$-satisfied, 
there are at least $\Theta(\alpha^*)$ students whose budgets are either at most $1$ or at least $1+\epsilon'$.
Thus if the assignment to the Course Allocation problem does not correspond 
to an $(\epsilon, \delta)$-approximate solution to {\sc Gcircuit},
then a $\delta'=\Theta(\delta/\log(1/\epsilon))$-fraction of the students must have budgets at most $1$ or at least $1+\epsilon'$.
Applying Lemma \ref{lem:gini} completes the proof of Theorem \ref{thm:ceei-informal}.
\qed

\end{document}